\documentclass[conference]{IEEEtran}

\usepackage{amssymb}
\usepackage{amsthm}
\usepackage{mathtools}
\usepackage{subcaption}
\usepackage{hyperref}
\usepackage{enumitem}

\usepackage[usestackEOL]{stackengine}

\usepackage{graphicx}       % graphics
\graphicspath{{gfx/}}

\usepackage{orcidlink} % to get links to orcid-profile

\providecommand{\keywords}[1]{\textbf{\textit{Keywords---}} #1}

\providecommand{\pgfsyspdfmark}[3]{}

\newcommand{\sgn}{\text{sgn }}
\DeclareMathOperator*{\argmin}{arg\,min}

\newcommand{\defeq}{\vcentcolon=}
\newcommand{\eqdef}{=\vcentcolon}

\theoremstyle{plain}
\newtheorem{theorem}{Theorem}[section]
\theoremstyle{plain}
\newtheorem{corollary}{Corollary}[theorem]
\theoremstyle{definition}
\newtheorem{lemma}{Lemma}[section]
\theoremstyle{definition}
\newtheorem{definition}{Definition}[section]
\theoremstyle{definition}

\begin{document}
\title{NISQ-ready community detection based on separation-node identification}

%\author{
%\IEEEauthorblockN{Jonas Stein\IEEEauthorrefmark{1}, Dominik Ott\IEEEauthorrefmark{1}, %Mirco Schoenfeld\IEEEauthorrefmark{2} and Sebastian Feld\IEEEauthorrefmark{3}}
%\and\and
%\IEEEauthorblockA{\hspace{1cm}\IEEEauthorrefmark{1}
%{Mobile and Distributed Systems Group}\\
%{LMU Munich} \\
%{Munich, Germany} \\
%jonas.stein@ifi.lmu.de, d.ott@campus.lmu.de}
%\and
%\IEEEauthorblockA{\IEEEauthorrefmark{2}
%{Interdisciplinary Data Management \& Knowledge Generation}\\
%{University of Bayreuth}\\
%{Bayreuth, Germany} \\
%mirco.schoenfeld@uni-bayreuth.de}
%\and
%\IEEEauthorblockA{\IEEEauthorrefmark{3}
%{Quantum \& Computer Engineering}\\
%{Delft University of Technology}\\
%{Delft, The Netherlands} \\
%s.feld@tudelft.nl}
%}
\author{Jonas Stein$^{1}$\orcidlink{0000-0001-5727-9151}, Dominik Ott$^{1}$, Jonas Nüßlein$^{1}$\orcidlink{0000-0001-7129-1237}, David Bucher$^{2}$\orcidlink{0009-0002-0764-9606}, Mirco Schoenfeld$^{3}$\orcidlink{0000-0002-2843-3137} and Sebastian Feld$^{4}$\orcidlink{0000-0003-2782-1469}\\
	\normalsize $^{1}$Mobile and Distributed Systems Group, LMU Munich\\
    \normalsize $^{2}$Aqarios GmbH\\
	\normalsize $^{3}$Data Modelling \& Interdisciplinary Knowledge Generation, University of Bayreuth\\
    \normalsize $^{4}$Delft University of Technology\\
	\normalsize e-mail: jonas.stein@ifi.lmu.de
}

\maketitle

\begin{abstract}
The analysis of network structure is essential to many scientific areas, ranging from biology to sociology. As the computational task of clustering these networks into partitions, i.e., solving the community detection problem, is generally NP-hard, heuristic solutions are indispensable. The exploration of expedient heuristics has led to the development of particularly promising approaches in the emerging technology of quantum computing. Motivated by the substantial hardware demands for all established quantum community detection approaches, we introduce a novel QUBO based approach that only needs number-of-nodes many qubits and is represented by a QUBO-matrix as sparse as the input graph's adjacency matrix. The substantial improvement on the sparsity of the QUBO-matrix, which is typically very dense in related work, is achieved through the novel concept of separation-nodes. Instead of assigning every node to a community directly, this approach relies on the identification of a separation-node set, which -- upon its removal from the graph -- yields a set of connected components, representing the core components of the communities. Employing a greedy heuristic to assign the nodes from the separation-node sets to the identified community cores, subsequent experimental results yield a proof of concept. This work hence displays a promising approach to NISQ ready quantum community detection, catalyzing the application of quantum computers for the network structure analysis of large scale, real world problem instances.
\end{abstract}

%\begin{IEEEkeywords}
%Quantum Computing, Community Detection, QUBO, NISQ
%\end{IEEEkeywords}


\IEEEpeerreviewmaketitle

\section{Introduction}
In the era of digitization, the amount of collected data is rising rapidly. This poses substantial problems in data analysis as the algorithms employed there typically have superlinear and thus deficient runtime for many relevant datasets. In this work, we investigate a new approach to cope with this problem in the domain of graph structure analysis. Graphs are one of the central data structures used in information theory and find application in a vast range of scientific disciplines \cite{Bondy1976, Mashaghi2004, Shah2019}. The task of identifying the inherit structure of a graph is known as community detection \cite{FORTUNATO201075}. In practice, the use of corresponding clustering methods allows for the discovery of structural information from real world networks in domains ranging from social science to biology \cite{Girvan2002, Fani2017}.

Although no exact definition has been agreed upon, a graph is typically said to inherit a \textit{community structure} if it can be partitioned in a way such that the number of edges within the partitions is higher than the number of edges between the partitions \cite{Girvan2002}. While some approaches exist that can provably find existing community structures, all of them are NP-hard \cite{Nadakuditi2012, Brandes2006, Decelle2011, Newman2016}. This indicates a general NP-hardness of community detection and hence poses a demand for efficient heuristics to acquire solutions in reasonable time. Motivated by recent advancements and promising results in solving NP-hard problems in the field of quantum computing (QC) \cite{Arute2019, Shaydulin2019, Denchev2016, Albash2018}, we investigate possible advantages in building such heuristics by utilizing the more powerful algorithmic toolset available in QC.

In general, quantum computers allow for the application of quantum mechanical effects to do computation. Based on the concepts of superposition and entanglement, quantum computers can solve many computational problems provably faster than classical computers \cite{Grover1996, Shor1997, Lloyd1996}. In the case of community detection, related work has shown promising results using the popular modularity maximization approach \cite{Shaydulin2019, Ushijima-Mwesigwa2017}. Modularity is a measure for the quality of a given partitioning based on comparing the edge distribution of the given graph to the edge distribution of a graph with the same node degree but inheriting no community structure \cite{Newman2004}. The more these distributions differ, the higher the modularity indicating a clearer community structure. While this approach is provably optimal in the sense that no other approach could detect a community structure when modularity maximization cannot \cite{Nadakuditi2012}, its implementation on a quantum computer is cumbersome, especially for current quantum computers.

Present implementations of modularity maximization on quantum computers make use of the intrinsic quadratic nature of the modularity \cite{Shaydulin2019, Ushijima-Mwesigwa2017}. Simulating the time evolution of a specific quantum physical system, i.e., typically the transverse field Ising Model under adiabatic time evolution, a quantum heuristic solver for quadratic unconstraint binary optimization (QUBO) problems (e.g., modularity maximization) can be implemented on a quantum computer \cite{Shaydulin2019, Kadowaki1998}. Even though no quantum speedups where proven for solving NP-hard optimization problems with this approach jet, many cases of potential scaling advantages have been identified, with modularity maximization being one of them \cite{Shaydulin2019, Denchev2016, Albash2018}.

A critical limitation of the established quantum modularity maximization approach hindering its execution on near term quantum hardware is the size of the search space in optimization. Scaling linearly in the number of nodes and the number of communities, the required amount of quantum bits (qubits) needed for representing a specific solution quickly exceeds the number of qubits available in present noisy intermediate scale quantum (NISQ) hardware \cite{Preskill2018}.

Motivated by these results, we develop a novel approach to community detection, specialized for (quantum heuristic) QUBO solving that uses a smaller search space than the state-of-the-art quantum modularity maximization approach. This objective led to the sociologically inspired approach of defining a community by its extreme ends, similar to, e.g., differentiating political parties by their position on the left-right spectrum. For graphs, we translate this idea to the existence of, what we will later define as, a bijective set of separation nodes. The removal of the nodes contained in this set then yields connected components, which represent the "cores" of the communities. We subsequently conduct experiments, that indicate that this essentially solves the hard part of the community detection problem, as the community assignment for the separation nodes can typically be obtained using a greedy optimizer.

This idea allows for a quantum-classical hybrid detection of communities while merely using one qubit for every node in the graph with a single call to a QUBO solver. We show empirically that such a set of separation nodes can be found for graphs inheriting community structure and introduce a quantum heuristic approach to find them, constituting a proof-of-concept.

This paper is structured in the following way: in section \ref{sec:relatedwork} we describe the current state of the art of quantum community detection, in section \ref{sec:concept} the separation node set approach to (quantum) community detection is introduced to then get evaluated in section \ref{sec:evaluation} before concluding the findings in section \ref{sec:conclusion}.

%\PARstart{F}{or} papers submitted to virtual special issues, the editor will provide a footnote for the bottom of the first page that must be reproduced verbatim.\footnote{Special issue on a specific very important topic.}

\section{Related Work}
\label{sec:relatedwork}
With the advent of quantum optimization heuristics like quantum annealing, possible quantum advantages have been explored for many optimization problems \cite{Dalyac2021}. Easily allowing for a binary encoding of solutions and showing promising performance, community detection quickly became a popular problem in quantum optimization \cite{AKBAR2020}.

Representing community detection natively as a QUBO problem in the basic case of partitioning into $k=2$ communities, modularity maximization was the first approach used in quantum community detection \cite{Ushijima-Mwesigwa2017}. For a given graph $G=\left(V,E\right)$, the modularity of a partitioning into $V_0=\left\lbrace v_i \in V\mid x_i=0\right\rbrace$ and $V_1=\left\lbrace v_i \in V\mid x_i=1\right\rbrace$ according to $x=\left(x_1,...,x_{\left|V\right|}\right) \in \left\lbrace 0, 1\right\rbrace^{\left|V\right|}$ is given by
\begin{equation}
    \dfrac{1}{2\left|E\right|}\sum_{ij}\left(a_{ij}-\dfrac{d_{i}d_{j}}{2\left|E\right|}\right)x_i x_j
\label{eq:mod_max}
\end{equation}
for given node degrees $d=\left(d_1,...,d_{\left|V\right|}\right)$ and $a_{ij}$ denoting the entries of the adjacency matrix $A$ of $G$. Straightforward calculations yield the resulting QUBO matrix $Q=A-\frac{dd^\intercal }{2\left|E\right|}$ which is sufficient to apply practically all currently available quantum optimization heuristics.

This approach to can be generalized to $k>2$ communities by introducing one-hot encoding. Here, the community assignment of a node $v_i \in V$ is encoded by a $k$-dimensional bit string $x_{i}=\left(x^{\left(1\right)}_i,..., x^{\left(k\right)}_i\right)$ with $x^{\left(l\right)}_i=1$ and $x^{\left(m\right)}_i=0$ $\forall m\neq l$ if the node $v_i$ is assigned to community $l$. The resulting optimization term is hence given by:
\begin{equation}
    \dfrac{1}{2\left|E\right|}\sum_{ij}\left(a_{ij}-\dfrac{d_{i}d_{j}}{2\left|E\right|}\right)\left(\sum_l x^{\left(l\right)}_i x^{\left(l\right)}_j\right)
\label{eq:mod_max_arbitrary_k}
\end{equation}
In order to formulate this as a QUBO problem, we have to add a suitably weighted penalty term $P(x)$ (for details, see \cite{Zahedinejad2020}) to the optimization term to indirectly enforce the one-hot encoding by $P(x)=0$ if every node is assigned to exactly one community and $P(x)>0$ otherwise:
\begin{equation}
    P(x)=\sum_i \left(1-\sum_l x^{\left(l\right)}_i\right)^2
\end{equation}

Apart from capitalizing on the inherent QUBO form of modularity maximization, many other quantum computing based approaches to community detection like Quantum Genetic Algorithms and Quantum Walks have been proposed in recent literature \cite{Sedghpour2017, Mukai2020}. A particularly promising approach for near term application on large graphs is based on exploiting the quadratic nature of regularity checking related to Szemeredi's Regularity Lemma (SRL) \cite{Reittu2019}. While similar to our approach, as the QUBO problems solved only involve $\left|V\right|$ qubits, it fundamentally works differently, as communities are identified iteratively. In essence, the algorithm proposed in \cite{Reittu2019} executes the following steps:
\begin{enumerate}
    \item Randomly split the given graph $G=\left(V,E\right)$ into two equally sized partitions $A \dot{\cup} B = V$ and delete all edges inside the partitions to yield a bipartite graph
    \item Find subsets $X\subseteq A$ and $Y\subseteq B$ such that $X=\left\lbrace v_{i}\in A \mid s_i = 1\right\rbrace$ and $Y=\left\lbrace v_{j}\in B \mid s_j = 1\right\rbrace$ where $s=\left(s_1,...,s_{\left|V\right|}\right)$ is the solution to the quadratic program given by:
    \begin{equation}
        \argmin_{s\in\left\lbrace 0,1\right\rbrace ^{\left|V\right|}} \sum_{\substack{v_{i}\in A \\ v_{j} \in B}} \left(d(A,B) - a_{ij}\right) s_i s_j
    \end{equation}
    Here, $d(V_1,V_2)$ denotes the link density of two disjoint sets $V_1$, $V_2$ given by $\frac{e(V_1,V_2)}{\left|V_1\right|\left|V_2\right|}$ and $e(V_1,V_2)$ represents the number of edges connecting $V_1$ and $V_2$.
    \item Identify $C\defeq X\cup Y$ to be a community and repeat steps 1) and 2) for the subgraph induced on $G$ by $V\setminus C$.
\end{enumerate}
While this approach has a solid graph theoretic foundation, the high number of needed solver calls and the dense QUBO matrix still pose nontrivial hardware execution challenges in the NISQ era.

Aiming to minimize the demands to the QUBO solver, we propose a radically different approach that only needs a single QPU call and whose QUBO matrix is topologically identical to adjacency matrix of the given graph and hence, equally sparse. The approach presented in this work essentially purifies a solution of relaxed community detection problem, i.e., the final community structure is represented by the solution of a QUBO problem which is based on classically computed, probabilistic community assignments for each node. While we introduce a particularly efficient approach to calculate the needed input for the QUBO problem, many other approaches to relaxed community detection have been proposed in related work like semidefinite programming or convexification \cite{Chan2011, Chen2018, Abdalla2022, LI2015}.

As derived in detail in the next section, our approach requires a solution for a novel relaxation of the community detection problem as input to the QUBO problem formulation. In essence, our approach demands for an estimate value for each edge, specifying whether it connects nodes belonging to different or the same communities. While such estimates could in principle be computed based on the output of solvers for the relaxed community detection problem by using, e.g., the KL-divergence of the community affiliations of neighboring nodes, we introduce a specialized estimation method tailored to this task. Notably, metrics like the edge betweenness centrality \cite{Brandes2001} also do not yield satisfactory results for our approach, as the difference in values between separation- and non-separation-edges is seemingly too small.

\section{Concept}
\label{sec:concept}
In the following, we explore the idea of performing community detection based on finding a suitable set of nodes separating the communities as defined in \ref{def:sepnodeset} in a rigorous mathematical manner. Meeting the demand from the derived QUBO formulation for a separation edge estimator, we subsequently introduce a promising heuristic approach based on the concept of modularity.

\subsection{Separation-node sets}
\label{subsec:sep-node-sets}
The approach presented in this paper consists of two steps:
\begin{enumerate}[label=\textnormal{(\arabic*)}]
	\item \phantomsection\label{itm:one} identifying a set of nodes separating communities and thus revealing the fundamental community structure (sec. \ref{subsec:mod-based-estimation} and \ref{subsec:edge-nc}) 
	\item \phantomsection\label{itm:two} classifying the community of each separation-node to finalize the community detection (sec. \ref{subsec:partitioning})
\end{enumerate}
Either using a trivial, greedy approach introduced in \ref{subsec:partitioning} or a slight adaptation of the well-known QUBO-formulation of modularity maximization \cite{Negre2020} to perform \ref{itm:two}, the main objective of this paper is the development of a QUBO-approach realizing \ref{itm:one}. To provide a more formal definition of \ref{itm:one}, we now introduce the concept of \textit{separation-node sets}. In the following, we use $\mathcal{S}$ to denote the set of all separation-node sets.
\begin{definition}
For a graph $G=\left(V,E\right)$ and a ground truth community structure $C$ partitioning $V$, we call $S\subseteq V$ a \textbf{set of separation-nodes} iff the connected components $\overline{S}_i$ partitioning the graph induced by $V\setminus S$ are distributed such that $\left\{\overline{S}_i\right\}_i$ is a refinement of $C$.
\label{def:sepnodeset}
\end{definition}
Equivalent to this definition, one could also demand the existence of a refinement map $\phi :\mathcal{P}\left(V\right)\rightarrow\mathcal{P}\left(V\right)$ mapping each connected component $\overline{S}_i\subseteq V$ onto a community $\phi\left(\overline{S}_i\right)=C_{j}\in C$ such that $\overline{S}_i\subseteq C_{j}$. Utilizing the notion of separation-node sets, \ref{itm:one} can be formulated as finding a smallest set of separation-nodes whose associated refinement map $\phi$ is ideally bijective. An example of a set of separation-nodes satisfying these conditions is depicted in subfigure \ref{subfig:separationNodeIdentification}, which is part of figure \ref{fig:approach} displaying the proposed approach. As it will become apparent in the evaluation, such well behaved separation node sets can also be found in graphs with application near topologies.

\begin{figure}[ht]
        \centering
        \begin{subfigure}[t]{0.22\textwidth}
            \centering
            \includegraphics[scale=1]{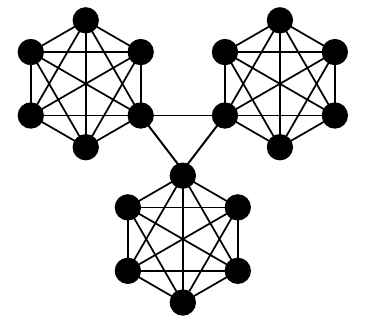}
            \caption{Exemplary graph consisting of three connected cliques.}
            \label{subfig:exampleGraph}
        \end{subfigure}
        \hfill
        \begin{subfigure}[t]{0.22\textwidth}  
            \centering 
            \includegraphics[scale=1]{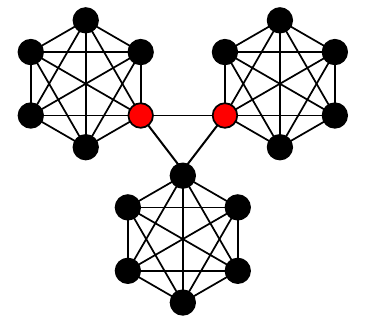}
            \caption{\small Identification of a set of separation-nodes (marked in red).}
            \label{subfig:separationNodeIdentification}
        \end{subfigure}
        \vskip\baselineskip
        \begin{subfigure}[t]{0.22\textwidth}   
            \centering 
            \includegraphics[scale=1]{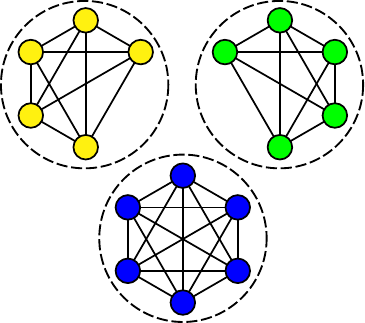}
            \caption{Removal of the set of identified separation-nodes and identification of the resulting connected components.}
            \label{subfig:connectedComponentsIdentification}
        \end{subfigure}
        \hfill
        \begin{subfigure}[t]{0.22\textwidth}   
            \centering 
            \includegraphics[scale=1]{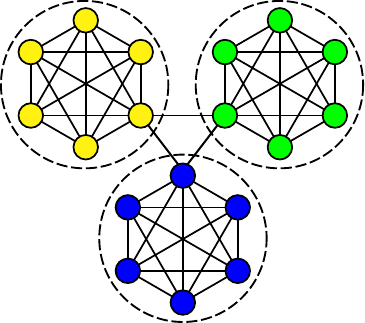}
            \caption{Classification of the community of all identified separation-nodes.}
            \label{subfig:separationNodeClassification}
        \end{subfigure}
        \caption{Outline of the workflow for the approach proposed approach of community detection via separation-node identification. The computationally expensive tasks of identifying a set of separation-nodes (subfigure \ref{subfig:separationNodeIdentification}) and classifying the communities for these nodes (subfigure \ref{subfig:separationNodeClassification}) are performed using quantum computing, while the computationally cheap tasks of removing the classified separation-nodes and identifying the resulting connected components (subfigure \ref{subfig:connectedComponentsIdentification}) are done classically.}
        \label{fig:approach}
    \end{figure}
The surjectivity of $\phi$ ensures that each community gets detected and its injectivity ascertains, that no communities get split. In the following, we will call separation-node sets injective, surjective or bijective iff the respective refinement function satisfies these conditions. In order to formulate a QUBO problem where the optimal solution represents the minimal separation-node set, we start by stating an alternate, more convenient, definition of minimal separation-node sets.

\begin{theorem}
\label{th:minimality}
For an adequate penalty term $P$ ensuring the separation-node set properties, the following equation states an equivalent definition of the set containing all minimal separation-node sets $\mathcal{S}_{min}$.
\begin{equation}
\mathcal{S}_{min}=\left\{\bigcup_{\substack{v_{i}\in V \\ x_{i}=0}}v_{i} \mid x=\argmin_{x\in\left\{0,1\right\}^{\left| V\right|}} 2P(x) -\sum_{v_{i}\in V} x_{i} \right\}
\label{eq:S-min}
\end{equation}
Here, we used $x\in\left\{0,1\right\}^{\left| V\right|}$ as a $0$-flag for separation-nodes, $a_{ij}$ to denote the entries of the adjacency matrix,  $c: V \rightarrow C$ as a mapping of nodes to their ground truth community and the Kronecker delta $\delta_{xy}$. For a penalty term $P$ ensuring the validity of the separation-node set definition by penalizing incident node pairs from strictly different communities where neither node is element of the sought-after separation-node set, see the following definition:
\begin{equation*}
P\left(x\right)\defeq \sum_{\left(v_{i},v_{j}\right)\in V^{2}}a_{ij}\left(1-\delta_{c\left(v_{i}\right)c\left(v_{j}\right)}\right) x_{i}x_{j}
\end{equation*}
\end{theorem}
\begin{proof}
See appendix \ref{subsec:apx-th-IV1}.
\end{proof}
Therefore, the task of finding a smallest set of separation-nodes for any given graph is native to the concept of QUBO and its formulation can be reduced to approximating $\delta_{c\left(v_{i}\right)c\left(v_{j}\right)}$ for incident node pairs $v_{i},v_{j}\in V$. This can be understood as calculating the probability of an edge being an interconnection of adjacent nodes belonging to different communities, or more formally, a \textit{separation-edge}.

Most interestingly, we can show that solving the QUBO problem stated in equation \ref{eq:S-min} is NP-hard for a specific estimator. To see this, we start by observing a substantial similarity of our QUBO formulation with the QUBO formulation of the Max-Clique problem as stated in \cite{Chapuis2019}:
\begin{equation}
    \argmin_{x\in\left\{0,1\right\}^{\left| V\right|}} 2\sum_{\left(v_{i},v_{j}\right)\in V^{2}} \left( 1- a_{ij}\right) x_i x_j -\sum_{v_{i}\in V} x_{i}
\end{equation}
for a given graph $G=(V,E)$ and its corresponding adjacency matrix $A$ with entries $a_{ij}$. Choosing the estimator $s:V\times V\rightarrow{} \left\lbrace 0,1\right\rbrace$ by $s\left(\left(v_i, v_j\right)\right)\defeq a_{ij}$, it becomes apparent, that the QUBO formulations are identical if we specify to use a complete graph of size $\left|V\right|$ as an input to our QUBO formulation. Leaving an extensive mathematical analysis of the NP-hardness for more realistic estimators to future work, this shows that the problem of finding a minimal separation-node set is NP-hard when treating the estimator as a variable. This result supports the pursuit of the proposed approach of using quantum computing in order to find a minimal separation-node. 

Returning to the initial goal of finding bijective separation-node sets, we now explore their surjectivity. A significant discovery regarding surjectivity is illustrated in figure \ref{fig:ce-surj}, showing no-free-lunch when using theorem \ref{th:minimality} to find surjective separation-node sets. This necessitates the addition of a penalty term to the QUBO formulation in order to ensure surjectivity when building upon theorem \ref{th:minimality}. For the formulation of a suitable penalty term, see appendix \ref{subsec:apx-insur}.

\begin{figure}[ht]
\begin{subfigure}[t]{.14\textwidth}
\centering
\includegraphics[width=\linewidth]{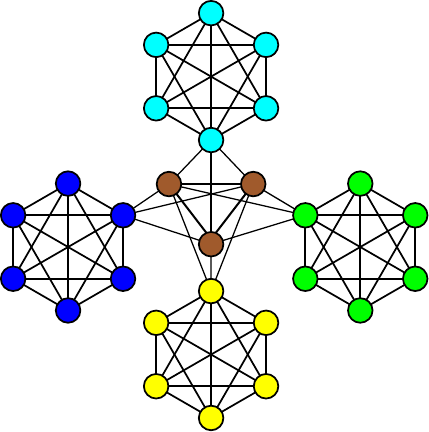}
\caption{A graph and its community structure.}
\end{subfigure}\hfill
\begin{subfigure}[t]{.14\textwidth}
\centering
\includegraphics[width=\linewidth]{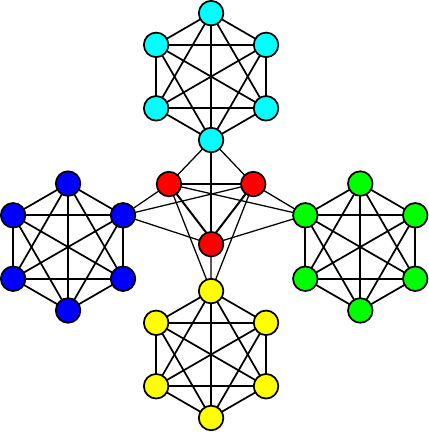}
\caption{Minimal separation-node set consisting of all nodes marked in red.}
\end{subfigure}\hfill
\begin{subfigure}[t]{.14\textwidth}
\centering
\includegraphics[width=\linewidth]{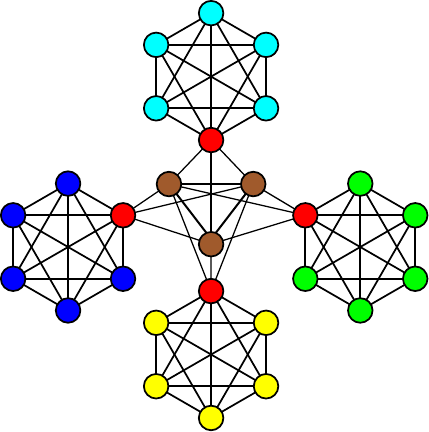}
\caption{Minimal surjective separation-node set consisting of all nodes marked in red.}
\end{subfigure}
\caption{Counterexample proving no-free-lunch when using theorem \ref{th:minimality} to find surjective separation-node sets.}
\label{fig:ce-surj}
\end{figure}

As our formulation results in a PUBO (polynomial unconstrained binary optimization) problem of degree $\mathcal{O}(\log_2 |V|)$, we conjecture, that this constraint cannot be realized in QUBO form without the addition of ancillary variables. Using the standard quadratization approach with the Rosenberg polynomial \cite{Rosenberg1975}, a QUBO formulation of this term demands superpolynomially many ancillary variables, i.e., $\mathcal{O}\left(\left|V\right|^{2\log_2 \log_2 \left|V\right|}\right)$. In the context of quantum annealing, this scaling beyond a quadratic number of qubits makes the surjective separation-node approach overly complex compared to the standard modularity maximization. In the gate model, the QAOA can be used to solve PUBO problems in principle, but as current hardware limitations prohibit adequate evaluation, we leave the exploration of the surjectivity constraint to future work.

As a consequence of not enforcing surjectivity, there exists a possibility that number of communities is incorrect after step \ref{itm:one} of detecting the fundamental community structure by separation-node set identification. Modifying step \ref{itm:two} slightly, this could in principle be compensated by iteratively increasing the number of possible communities until no further improvement of the modularity can be achieved. A clever way to do this, could be the elbow-method as known in clustering \cite{Thorndike1953}. For the alternative greedy approach for the second step \ref{itm:two}, the possibility of merging communities could be allowed.

Fortunately, conducted experiments show that topological structures precluding free lunch are scarce in practice. Therefore, we will omit the explicit demand for surjective separation-node sets in the following.

Analog to the surjectivity, there exist graph topologies like the one displayed in figure \ref{fig:ce-inj} showing no free lunch when using theorem \ref{th:minimality} to find injective separation-node sets. Hence, it appears necessary to ensure injectivity explicitly using a penalty term when building upon theorem \ref{th:minimality} in principle, as well. The formulation of such a penalty term also turns out to be rather tedious, as can be seen in appendix \ref{lem:inSignTerm}. In this case, we end up with an even higher dimensional PUBO problem for the injectivity than for the surjectivity. Luckily, compared to the surjectivity, the injectivity of a separation-node set is of less importance, as the second step \ref{itm:two} could easily be adapted to cope with this. Analog to the case of surjectivity, we observe such topological structures preventing free lunch quite rarely in conducted experiments, resulting in the analog dismissal of an explicit demand for the separation-node sets to be injective in practice.

\begin{figure}[ht]
\begin{subfigure}[t]{.14\textwidth}
\centering
\includegraphics[width=\linewidth]{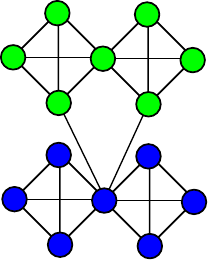}
\caption{A graph and its community structure.}
\end{subfigure}\hfill
\begin{subfigure}[t]{.14\textwidth}
\centering
\includegraphics[width=\linewidth]{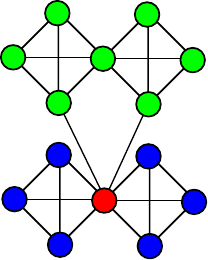}
\caption{Minimal separation-node set consisting of all nodes marked in red.}
\end{subfigure}\hfill
\begin{subfigure}[t]{.14\textwidth}
\centering
\includegraphics[width=\linewidth]{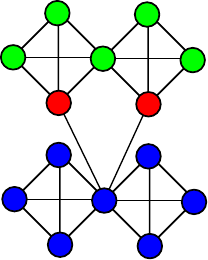}
\caption{Minimal injective separation-node set consisting of all nodes marked in red.}
\end{subfigure}
\caption{Counterexample indicating no-free-lunch when using theorem \ref{th:minimality} to find injective separation-node sets.}
\label{fig:ce-inj}
\end{figure}

In summary, the apparent infrequence of topological structures preventing free lunch regarding bijectivity renders the QUBO-formulation stated in theorem \ref{th:minimality} to be a well-founded starting point for the stated proposition of QUBO based community detection via separation-node sets.

While this approach provides exact results for a perfect classification of separation-edges, it fully relies on a suitable estimation heuristic. Although many known measures for a variety of edge properties exist (as described in \ref{sec:relatedwork}), none showed to be entirely suitable for detecting separation-edges according pretesting conducted for this paper. Consequently, we now motivate a novel approach tailored for exactly this task based on the concept of modularity.

\subsection{Modularity based separation-edge estimation}
\label{subsec:mod-based-estimation}
Motivated by the proven optimality of modularity and by the fact that at its core, modularity stems on essentially estimating whether each node pair is likely to belong to the same or different communities, we start by showing how this idea can be used to estimate $\delta_{c\left(v_{i}\right)c\left(v_{j}\right)}$. For this, recall the definition of the entries of the modularity matrix:
\begin{equation}
m_{ij} \defeq \dfrac{a_{ij} - \mathbb{E}\left[ J_{ij}\right]}{\left|E\right|}
\end{equation}
As before $a_{ij}$ are the entries of the respective adjacency matrix, while $\mathbb{E}\left[ e_{ij}\right]$ denotes the expected value of the number of edges between $v_i$ and $v_j$, $J_{ij}$. Upon closer inspection, we observe two main cases:
\begin{itemize}
\item $m_{ij}>0$, iff less connectivity between $v_i$ and $v_j$ was to be expected, indicating that $v_i$ and $v_j$ likely belong to the same community
\item $m_{ij}<0$, iff more connectivity between $v_i$ and $v_j$ was to be expected, indicating that $v_i$ and $v_j$ likely belong to different communities
\end{itemize}
As the matrix entries $m_{ij}$ are normalized to the interval of $\left[-1,1\right]$ by the division with $\left|E\right|$, we can see, that using proper rescaling to the intervall of $\left[0,1\right]$, i.e., via $2\left(m_{ij}+1\right)$, this allows for an estimation of the term $\delta_{c\left(v_{i}\right)c\left(v_{j}\right)}$ in principle.

In practice however, this approach yields extremely bad estimations, as only the entries $m_{ij}$ of the modularity matrix are relevant, that correspond to a given edge $\left(v_i,v_j\right)\in E$. For these, it quickly becomes apparent, that $m_{ij}$ is typically larger than $0$, making this exact idea infeasible in practice. These considerations motivate an adaptation of modularity for the estimation of separation-edges as proposed in the following.

\subsection{Edge neighborhood connectivity based separation-edge estimation}
\label{subsec:edge-nc}
Exploiting the mathematical structure of modularity for a straightforward separation-edge estimation, we now introduce a promising generalization of the previous approach, which we coin as the \textit{neighborhood connectivity} of an edge. Instead of merely taking the direct connection between two nodes into account (i.e., an edge), the neighborhood connectivity of an edge considers connections between the neighborhoods of the nodes. In this context, the neighborhood $N_r\left(v\right)$ of a node $v\in V$ is defined to be the set of nodes with a shortest path of length $r$ to $v$. 

Based on this idea, we can rephrase the basic case of our generalization, i.e., modularity, as merely counting the number of unique edges on paths of length $1$ between the $0$-neighborhoods $N_0\left(v_i\right)$ and $N_0\left(v_j\right)$ of the respective nodes $v_i$ and $v_J$. The here proposed generalization introduces the following two new notions:
\begin{enumerate}[label=\textnormal{(\arabic*)}]
    \item Consider connections between $r$-neighborhoods with radius $r\geq 0$
    \item Also consider paths of length $2$
\end{enumerate}
Stating this more precisely in mathematical form, we now define the neighborhood connectivity $\nu_{r}^{\left(l\right)}$ of an edge given a path length $l$, and a neighborhood size $r$:
\begin{equation}
    \nu_{r}^{\left(l\right)}\defeq \dfrac{a_{r}^{\left(l\right)}-\mathbb{E}\left(a_{r}^{\left(l\right)}\right)}{n_{r}^{\left(l\right)}}
\end{equation}
In this definition, $a_{r}^{\left(l\right)}$ denotes the number of unique edges contained in paths of length $l$ connecting the $r$-neighborhoods of the given nodes which do not involve nodes or edges contained by the $\left(r-1\right)$-neighborhoods (as this would result in possible double counting of edges). Analogously to the definition of modularity, $\mathbb{E}\left(a_{r}^{\left(l\right)}\right)$ denotes the expected value corresponding to $a_{r}^{\left(l\right)}$ and $n_{r}^{\left(l\right)}$ acts as a normalization factor denoting the highest possible number $a_{r}^{\left(l\right)}$ can assume.

These values can be calculated based on a simple breadth first search with depth $r$ iterating of the neighborhood layers while choosing $v_i$ and $v_j$ as starting nodes. As for the expected value calculation, the configuration model has shown to be an adequate choice (which is in line with modularity). For details on this, we refer to our implementation which can be made available upon request to the authors.

Our preferred method of combining the results into the neighborhood connectivity $\nu$ of a given edge based on all $\nu_{r}^{\left(l\right)}$ is the dot product with a weight vector $w$ with entries $w_{r}^{\left(l\right)}\in \mathbb{R}_{0}^{+}$ such that their sum equals $1$:
\begin{equation}
    \nu\defeq\sum_{r=1}^{d}w_{r}^{\left(1\right)}\nu_{r}^{\left(1\right)} + \sum_{r=0}^{d-1}w_{r}^{\left(2\right)}\nu_{r}^{\left(2\right)}
\end{equation}
As we know that the standard modularity value is of little use, we chose $w_{0}^{\left(1\right)}=0$. We consider the remaining weights as hyperparameters, for which $w_{0}^{\left(2\right)}=0.5=w_{1}^{\left(1\right)}$ have proven to be suitable values according to conducted experiments.

\subsection{Assigning the separation-nodes to communities}
\label{subsec:partitioning}
As stated in sec. \ref{subsec:sep-node-sets}, we propose two different approaches to assigning the separation-nodes to communities, i.e., (1) a greedy strategy and (2) modularity maximization. In the experiments conducted in this paper, the greedy strategy was mainly employed for all experiments. It works as follows:
\begin{enumerate}[label=\textnormal{(\arabic*)}]
    \item count the number of edges to every know community for each separation-node
    \item assign the node with the most edges to a single community to that community
    \item update the counts for every neighboring separation-node
    \item repeat steps two and three until every separation-node is properly assigned to a community
\end{enumerate}
This algorithm has a runtime of the number of separation nodes $S$ times the number of communities $\left|C\right|$, $\mathcal{O}\left(S\cdot \left|C\right|\right)$ and hence runs very efficiently.

As the results calculated based on the edge neighborhood connectivity did not always show reasonable quality to sensibly use this greedy optimizer, we chose to use the standard of modularity maximization for these special cases. Fortunately, the well known QUBO approach to this \cite{Ushijima-Mwesigwa2017}, can be easily adapted to our situation, i.e., by clamping the values of the known community assignments, where clamping is to be understood in the same way as it is used in quantum Boltzmann machines \cite{Amin2018}. This yields a QUBO problem of size $\mathcal{O}\left(S\cdot \left|C\right|\right)$, which often can be solved a lot quicker than the original problem, as $S<\left| V\right|$ in practice.

\section{Evaluation}
\label{sec:evaluation}
The evaluation aims at the examination of the validity of the following two claims:
\begin{enumerate}[label=\textnormal{(\arabic*)}]
    \item\phantomsection\label{itm:claim1}  the assignment of separation-nodes to their communities is computationally easy, given a good enough estimator
    \item\phantomsection\label{itm:claim2} neighborhood connectivity allows for proof of concept results
\end{enumerate}
As we will show in the following, both claims appear to be valid according to the conducted experiments.

For investigating claim \ref{itm:claim1}, we propose to check, if the greedy separation-node assignment as described in sec. \ref{subsec:partitioning} is sufficient to assign the nodes of a well behaved separation-nodes to the correct communities. If this approach is indeed sufficient to obtain (nearly) perfect solutions, we reason that the claim is most likely valid.

In order to eliminate the possibility of an insufficient separation-node set, we use a synthetic dataset with known community structure, allowing for the use of a perfect estimator for the separation-edges. In order to find a very good separation-node set, we utilize a simple simulated annealing approach to solve the associated QUBO as defined in \ref{eq:S-min}. Regarding the synthetic dataset, we choose the stochastic block model (SBM) \cite{Holland1983} which is a widely used tool for benchmarking in the realm of community detection. Aiming to achieve realistic results, we use a graph of size $\left| V\right|=250$, structured into seven equally sized communities with varying intra- and interconnections between the communities, resembling three different difficulties, according to the phase transition of community detection on SBMs (for details on the phase transition, see \cite{Decelle2011}). As it becomes apparent in the corresponding figure \ref{fig:3by3}, the greedy separation-node assignment indeed yields optimal or at least close to optimal results, indicating the validity of claim \ref{itm:claim1}.

\begin{figure}[ht]
\centering
\includegraphics[width=0.45\textwidth]{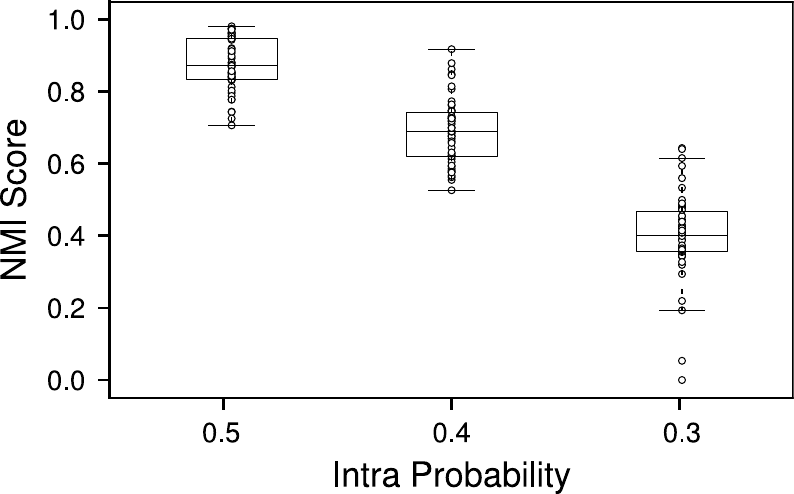}
\caption{This figure shows the NMI score of the presented approach for 50 different graphs each, based on ground truth and a perfect separation-edge estimator coupled with the greedy separation-node assignment. The NMI score as defined in \cite{Fred2003, Kuncheva2004} was used, as it resembles a well proven measure for the accuracy of a community given ground truth \cite{Danon2005}. The different probabilities for intra-community edges in the chosen SBM model resemble different difficulties according to the phase transition known for this model. The lower the stated probability, the harder the problem. The probabilities were chosen such that the hardest graphs barely differed from a null model inheriting no measurable structure up to the hardest that still allowed perfect NMI scores.}
\label{fig:3by3}
\end{figure}%ATTENTION: This figure will be edited to merely contain the NMI box plots displayed in the first row and nothing else.

Having seen solid results for the optimal estimator, we now want to investigate the performance of the here presented "neighborhood connectivity" approach for real world data and hence explore claim \ref{itm:claim2}. For this, we choose the greedy separation-node assignment, so that the separation-node identification displays the only non-trivial task, capable of solving the problem instances. Choosing standard real world benchmark graphs of varying size, we can observe stable results for most datasets in figure \ref{fig:real-world-data}, while often achieving 90 to 95\% optimal results.

\begin{figure}[ht]
\centering
\includegraphics[width=0.45\textwidth]{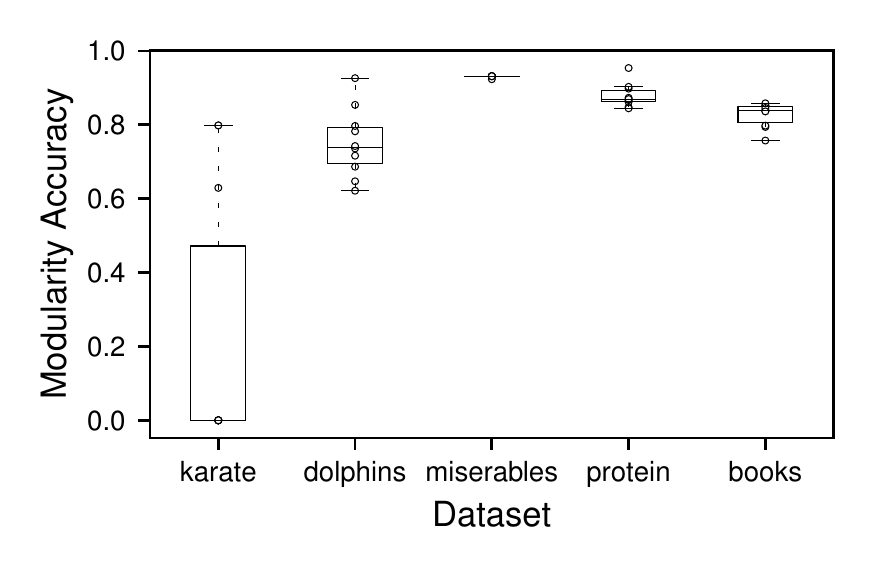}
\caption{This box plot displays the fraction of the achieved modularity score by the best known solution for selected standard benchmark datasets: (1) the social network of a karate club \cite{Zachary1977}, (2) the social interactions between dolphins \cite{Lusseau2003}, (3) the collectively appearing characters in the book "Les Miserables" \cite{Knuth1993}, (4) protein protein interacations \cite{Palla2005} and (5) jointly bought political books \cite{Newman2006}. Each graph was analyzed 10 times using simulated annealing. Our approach clearly does not work well for the karate club network. Closer inspections yield that the connected components resulting from the found separation-node sets often only consist of single nodes, indicating suboptimality in using neighborhood connectivtiy for this dataset.}
\label{fig:real-world-data}
\end{figure}

Motivated by these proof of concept results, we now investigate the performance of the proposed estimator (edge neighborhood connectivity) in order to explore its optimal mission scenario. For this, we again resort to equally formed SBM benchmark graphs with slightly higher intra-community connection probabilities, as they offer the comparison with ground truth information. Concretely, we chose these probabilities to be 0.75 for the easy case, 0.625 for the medium case and 0.5 for the hard case, which was the easy case for the experiments conducted with the perfect estimator (and picked previously as the hardest case to still yield perfect results).

%the next sentence references to boxplot-sur_inj_corr-2022-12-29--16-30-41-955690.pdf
Analogously to the perfect estimator, the identified separation-node sets were all valid and bijective in a small test run on 10 graphs. Switching the main optimization goal, we now examine the size of the identified separation-node sets for graphs of different difficulty, as displayed in figure \ref{fig:sep-node-set-size}. Here, we can see, that the sizes of the separation-node sets found are substantially larger than the best known solution, this becomes apparent especially for easy problem instances. Interestingly, the performance quality increases for harder problems in relative perspective, showing promising scaling behavior.

\begin{figure}[ht]
\centering
\includegraphics[width=0.45\textwidth]{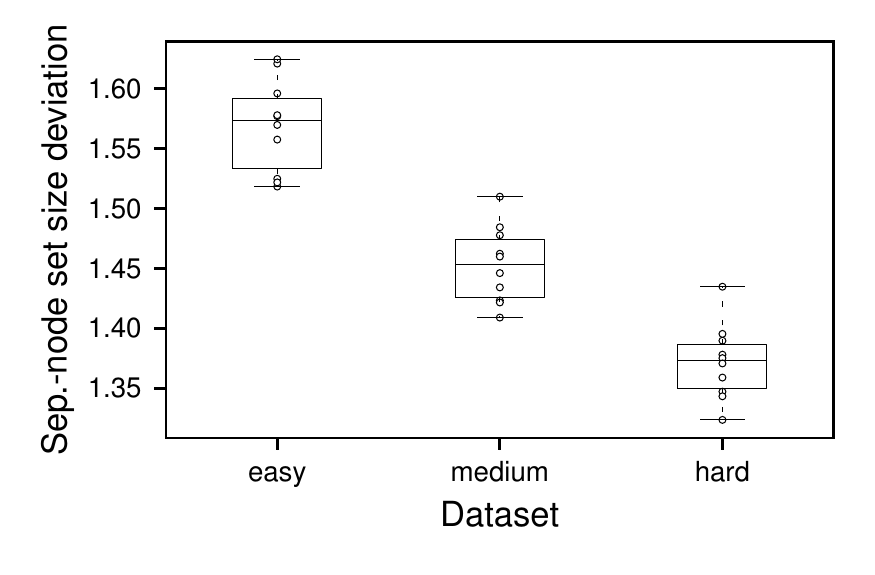}
\caption{The $y$-axis depicts the deviation factor from the best known separation-node set in size. Notably, the absolute sizes of the identified separation-node sets are typically similar over the different difficulties, while they rise slightly for larger graphs.}
\label{fig:sep-node-set-size}
\end{figure}

Although the separation-node sets found are well behaved, the combination with the greedy separation-node assignment to communities does yield substantially worse results than the perfect estimator, as shown in figure \ref{fig:LFR-NMI}.

\begin{figure}[ht]
\centering
\includegraphics[width=0.45\textwidth]{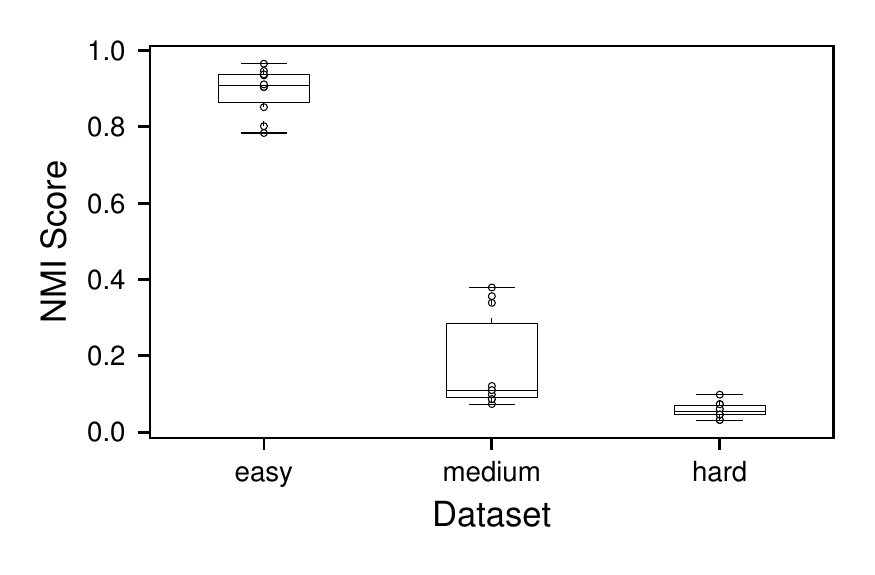}
\caption{This figure depicts the normalized mutual information score of the selected SBM benchmark graphs using the greedy assignment of separation-nodes to communities. A substantial drop off in performance can be observed for the harder datasets.}
\label{fig:LFR-NMI}
\end{figure}

Subsequent experiments show, that the performance for the medium and hard datasets can be improved significantly by exchanging the greedy approach for a simulated annealing based one, as shown in figure \ref{fig:LFR-ModMax-NMI}.

\begin{figure}[ht]
\centering
\includegraphics[width=0.45\textwidth]{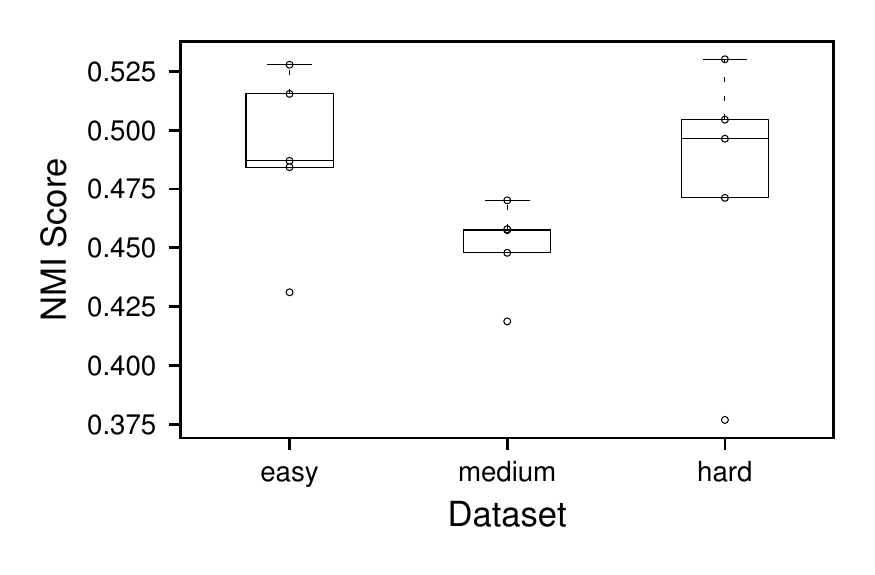}
\caption{This figure depicts the normalized mutual information score of the selected SBM benchmark graph using a simulated annealing based approach of assigning the separation-nodes to communities. The worse performance for the easy dataset clearly indicates that the chosen simulated annealing approach based on the QUBO as described in \ref{subsec:partitioning} is suboptimal in general.}
\label{fig:LFR-ModMax-NMI}
\end{figure}

As described in the caption of figure \ref{fig:LFR-ModMax-NMI}, simulated annealing based on the QUBO as described in \ref{subsec:partitioning} seems to be a suboptimal choice to assign separation-nodes to communities. We suspect that the reason for this resides in the large size of the search space for the given problem instances due to the employed one-hot encoding. As identified separation-node sets are typically sized up to 200 nodes (compared to the roughly 120 nodes for the perfect estimator), the search space for the problem instances thus contains roughly $\left(200\cdot 7\right)!=1400!$ possible solutions, as 7 different communities exist.

In order to put the results of the developed separation-edge estimator based on edge neighborhood connectivity into perspective with an optimal estimator, we now investigate its $R^2$ score in figure \ref{fig:R2-score}. Interestingly, the worse performance for larger datasets has no impact on the validty and bijectivity of the subsequently identified separation-node set, which is very promising in regards to scaling.

\begin{figure}[ht]
\centering
\includegraphics[width=0.45\textwidth]{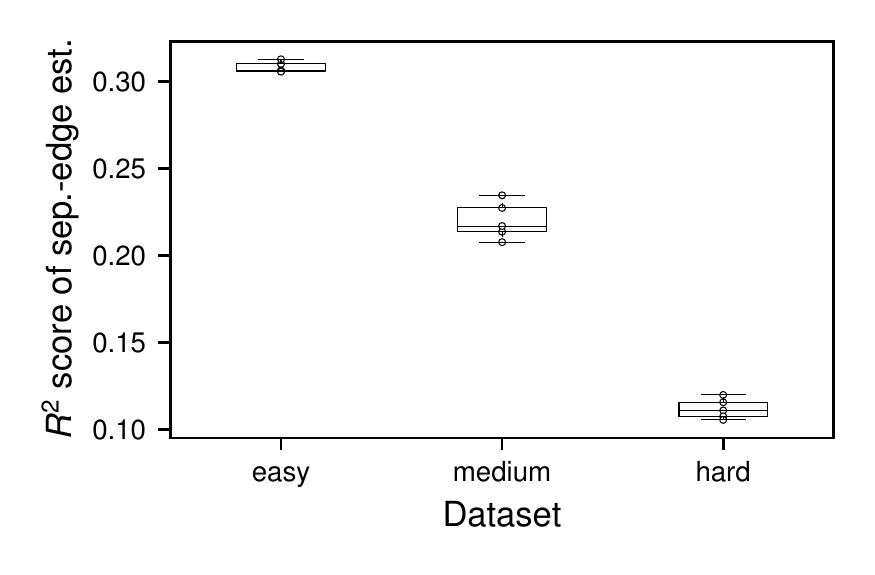}
\caption{$R^2$ score of the edge neighborhood connectivity based separation-edge estimator. In practice, an $R^2$ score of 30\% implies that merely 30\% of the variability of the ground truth has been accounted for. A strict trend towards worse results for harder datasets is clearly visible. This shows that the performance of the estimator decreases for harder problem instances as to be expected while still yielding somewhat accurate results.}
\label{fig:R2-score}
\end{figure}

\section{Conclusion}
\label{sec:conclusion}
Having set out with the goal of developing a quantum community detection approach that allows for the analysis of large graphs in the NISQ era, we presented the idea of identifying communities via their borders. The derived separation-node set based approach was shown to yield (close to) optimal results depending on the accuracy of the classical separation-edge estimator. The therefore proposed heuristic approach based on the introduced concept of "edge neighborhood connectivity" enabled for proof of concept results on real world data. In particular, as our approach merely requires $\left| V\right|$ qubits and as the corresponding QUBO is as sparse as the input graph $G=\left(V, E\right)$, separation-node based community detection resembles the least hardware demanding quantum computing approach to community detection to the best of our knowledge. The underlying trade-off necessary for this accomplishment clearly is the more demanding classical part in this hybrid approach (i.e., the separation-edge estimation). We firmly encourage future work on this heuristic, while conjecturing the incorporation of solutions of the relaxed community detection problem as highly beneficial. Furthermore, the exploration of adaptations of similar known metrics like edge betweenness centrality \cite{Li2017} also seem very interesting. Overall, we conclude our approach to be highly promising for accelerating the possibility of solving real world community detection problems using quantum computers and thus opening up a path towards network structure analysis in big data.

\section*{Acknowledgments}
This work was partially funded by the German BMWK project \textit{PlanQK} (01MK20005I) and the Bavarian StMWi project \textit{QAR-Lab Bayern} (46-6665e/443), respectively.

%Bibliography
\bibliographystyle{unsrt}  
\bibliography{main}  

\section*{Appendix}
\subsection{Proving theorem IV.1}
\label{subsec:apx-th-IV1}
In the following, we provide a proof for theorem \ref{th:minimality}, which states the following equation:
\begin{equation}
\mathcal{S}_{min}=\left\{\bigcup_{\substack{v_{i}\in V \\ x_{i}=0}}v_{i} \mid x=\argmin_{x\in\left\{0,1\right\}^{\left| V\right|}} 2P(x) -\sum_{v_{i}\in V} x_{i} \right\}
\end{equation}
Where, by definition, we have:
\begin{equation}
\mathcal{S}_{min}\defeq\left\{S \in \mathcal{S} \mid  \left|S\right| \leq \left|S'\right| \, \forall S'\in\mathcal{S} \right\}
\end{equation}
\begin{equation}
P\left(x\right)\defeq \sum_{\left(v_{i},v_{j}\right)\in V^{2}}a_{ij}\left(1-\delta_{c\left(v_{i}\right)c\left(v_{j}\right)}\right) x_{i}x_{j}
\end{equation}
Aiming to prove ``$\subseteq$'' and ``$\supseteq$'' individually, we first prove some lemmata.

\begin{lemma}
All $x\in\left\lbrace 0,1\right\rbrace^{\lvert V\rvert}$ satisfying $P(x)=0$ represent sets of separation-nodes.
\label{lem:Sep1}
\end{lemma}
\begin{proof}
Let $x$ be a binary vector such that $P(x)=0$ and let $S$ be the corresponding set of nodes. In order to prove the desired statement by contradiction, assume $S\notin \mathcal{S}$, which is equivalent to the existence of a connected component of the graph induced by $V\setminus S$ not being a subset of one community. Then at least two nodes $v_{i},v_{j}\in V$ must exist, that are connected via a path and belong to different communities. On this path, there must exist two adjacent nodes the belong to different communities with neither of them being an element of $S$. Therefore $P(x)$ must be bigger than $0$, yielding a contradiction.
\end{proof}

\begin{lemma}
The following equation states an alternative definition of the set containing all sets of separation-nodes.
\begin{equation}
\mathcal{S}=\left\{\bigcup_{\substack{v_{i}\in V \\ x_{i}=0}} v_i \mid P(x)=0 \right\}
\end{equation}
\label{lem:altSepNodeSetDef}
\end{lemma}
\begin{proof}
Using lemma \ref{lem:Sep1} to show  ``$\supseteq$'', we now show ``$\subseteq$''. Let $S\in\mathcal{S}$ be an arbitrary separation-node set and $x$ the corresponding binary vector $0$-flagging the nodes belonging to $S$. Assuming $P(x)\neq 0$, at least two adjacent nodes $v_{i},v_{j}\in V$ belonging to different communities exist following the definition of $P$. These nodes subsequently belong to the same connected component $\overline{S}_i$ of the graph induced by $V\setminus S$ implicating that no community can exist, that resembles a superset of the nodes inducing the connected component $\overline{S}_i$. Therefore $S$ can not be a set of separation-nodes as the corresponding refinement map can not exist, yielding a contradiction and showing $P(x)= 0$.
\end{proof}

\begin{lemma}
For every  $S\subset V$ satisfying $P(x)>0$, there exists a superset $\tilde{S}\supset S$ such that $\tilde{x}^{T}Q\tilde{x}<x^{T}Qx$, with $Q$ defined such that $x^{T}Qx=2P(x) -\sum_{v_{i}\in V} x_{i}$ and $\tilde{x}$ corresponding to $\tilde{S}$.
\label{lem:liftEmPenalties}
\end{lemma}
\begin{proof}
Let $S\subset V$ be a set of nodes such that the corresponding penalty term $P(x)$ is bigger than zero. This implicates the existence of a pair of incident nodes $v,w\in V$ being part of the same community while neither $v\in S$ nor $w\in S$. Then the set $\tilde{S}:=S\cup\left\{v\right\}$ (without loss of generality, we could also define $\tilde{S}:=S\cup\left\{w\right\}$ while achieving the same) has a smaller QUBO-value compared to $S$: With a decrement of at least $4$ in the the penalty term (i.e., taking its weighting of $2$ into account) and an increment in the cost function (i.e., the sum of the $x_i$'s) of $1$, we get $\tilde{x}^{T}Q\tilde{x}\leq x^{T}Qx - 1$, completing the proof.
\end{proof}

\begin{corollary}
$x=\argmin_{x\in\left\{0,1\right\}^{\left| V\right|}} 2P(x) -\sum_{v_{i}\in V} x_{i} \Rightarrow P(x)=0$.
\label{cor:PxZero}
\end{corollary}
\begin{proof}
This result  follows directly from application of lemma \ref{lem:liftEmPenalties}, as $P(x)\neq 0$ would violate the minimality property of $x$.
\end{proof}

With these lemmata, we are now ready to prove theorem \ref{th:minimality}:

\begin{proof}
Let $Q\in\mathbb{R}^{\left| V\right|\times \left| V\right|}$ be defined such that:
\begin{equation}
x^{T}Qx = 2P(x) - \sum_{v_{i}\in V} x_{i}
\label{eq:defQmin}
\end{equation}
We start by proving ``$\subseteq$'': Let $S\in \mathcal{S}_{min}$ and $x$ its corresponding $0$-flag vector, then we know by corollary \ref{cor:PxZero}, that $P(x)=0$. Therefore $x^{T}Qx = \left| S\right| - \left| V\right|\eqdef \overline{s_{min}}$. It is sufficient to show, that $\overline{s_{min}} = \min_{x\in\left\{0,1\right\}^{\left| V\right|}} -\sum_{v_{i}\in V} x_{i} + 2P(x)$. For this, we assume, that there exists an $\widetilde{x}$ such that $\widetilde{x}^{T}Q\widetilde{x}<\overline{s_{min}}$. Now, as $P\rightarrow \mathbb{N}_{0}$ two possibilities exist:
\begin{enumerate}
\item $P(\widetilde{x})=0$ and the separation-node set $\widetilde{S}$ is smaller than $S$.
\item $P(\widetilde{x})>0$ and the separation-node set $\widetilde{S}$ is much smaller than $S$.
\end{enumerate}
As we can see using lemma \ref{lem:liftEmPenalties}, we can reduce the latter case to the former case by iteratively eradicating all penalties. Now, using corollary \ref{cor:PxZero}, $\widetilde{S}$ is a separation-node set and by definition of $\widetilde{S}$, $\left|\widetilde{S}\right|<s$ yielding a contradiction to the minimality of $\mathcal{S}_{min}$ and thereby proving ``$\subseteq$''.

We now prove ``$\supseteq$'': Let $x^{*}\defeq \argmin_{x\in\left\{0,1\right\}^{\left| V\right|}} 2P(x)-\sum_{v_{i}\in V} x_{i}$ and let $S^{*}$ be the node set corresponding to $x^{*}$. As we can see using lemma \ref{lem:liftEmPenalties}, $P(x^{*})$ must be zero, otherwise $x^{*}$ could not be minimal in the sense of satisfying its definition. Therefore, $S^{*}$ is a separation-node set according to lemma \ref{lem:altSepNodeSetDef}. Assuming $S^{*}\notin\mathcal{S}_{min}$ yields $\left|S^{*}\right|\neq \left|S\right|$, for an arbitrary $S\in\mathcal{S}_{min}$ and thus two cases are possible:
\begin{enumerate}
\item $\left|S^{*}\right|< \left|S\right|$.
\item $\left|S^{*}\right|> \left|S\right|$.
\end{enumerate}
The former yields a contradiction to $\mathcal{S}_{min}$ being minimal and the latter yields a contradiction to the minimality of $x^{*}$.
\end{proof}

\subsection{Constructing penalty terms for the in- and surjectivity constraints}
\label{subsec:apx-insur}
In this section, we formulate penalty terms realizing the in- and surjectivity constraints for separation-node sets. Instead of solving this seemingly non-straightforward task directly, we will realize the individual constraints with terms, that are larger than $0$ iff the constraint is satisfied and $0$ otherwise. Exploiting the possibilities of PUBO, we will show that the respective terms can be used to build penalty functions using a moderate amount of ancillary variables.

\begin{lemma}
The separation-node set associated with the $0$-flag vector $x\in\left\lbrace 0,1\right\rbrace ^{\left|V\right|}$ is surjective iff $\sum_{v_j\in V} \delta_{c(v_i)c(v_j)}x_j>0$ for all $v_i\in V$.
\label{lem:surSignTerm}
\end{lemma}
\begin{proof}
    This equivalence can be trivially observed by closely inspecting the sum and is left as an exercise to the reader.
\end{proof}

As the heuristic approaches presented in this work only allow for the estimation of $\delta_{c(v_i)c(v_j)}$ for adjacent node pairs $\left(v_i,v_j\right)\in E$, no direct estimation for $\delta_{c(v_i)c(v_j)}$ seems accessible. However, we can use the estimation of $\delta_{c(v_i)c(v_j)}$ for adjacent node pairs, to estimate $\delta_{c(v_i)c(v_j)}$ for non-adjacent node pairs:
\begin{equation}
\delta_{c(v_i)c(v_j)}=\sgn \sum_{\overline{p}\in p(v_i,v_j)} \prod_{k=1}^{dim(\overline{p})-1}\delta_{c(\overline{p}_k)c(\overline{p}_{k+1})}
\end{equation}

Here, $p(v_i,v_j)$ denotes a function, that returns the set of all simple paths $(v_{\pi(1)},...,v_{\pi(l)})$ between $v_i$ and $v_j$. To allow for this convenient notation, we introduced $\pi : V \rightarrow \left\lbrace 1, ..., \left| V\right|\right\rbrace$ as the projection, mapping the indices of the path entries of the elements of $p(v_i,v_j)$ to their global indices from $\{v_1,...,v_n\}\in V$. In practice, these paths could be found using techniques presented in \cite{Sedgewick2001}.

The value inside sign function resembles the number of simple, intracommunity paths between $v_i$ and $v_j$. A general upper bound for the number of paths with these properties is $2^{\left| V\right| - 2}$, i.e., the number of all subsets of $V$ containing $v_i$ and $v_j$. For practical purposes, this term clearly is unsuitable, as small errors in the estimation of $\delta_{c(v_i)c(v_j)}$ add up very quickly. Neglecting applicability concerns for reasons described in \ref{subsec:sep-node-sets}, we now show how to build a PUBO penalty function associated to the surjectivity term used in lemma \ref{lem:surSignTerm}.

\begin{lemma}
Given a function $f: \left\{0,1\right\}^n \rightarrow \left\{0,...,m\right\}$ for arbitrary $n,m\in\mathbb{N}$ representing a constraint via $f(x)>0$, the following penalty terms can be used to ensure that $f(x)>0$ in PUBO:
\begin{align}
\begin{split}
P_1(x,y)\defeq & \left(f(x)-\sum_{i=0}^{\left\lceil \log_2(m)\right\rceil}2^i y_i\right)^2\\
   = & \begin{cases}
      0 & \text{if $\mathrm{y}\defeq \displaystyle \sum_{i=0}^{\left\lceil \log_2(m)\right\rceil}2^i y_i = f(x)$}\\
      >0 & \text{otherwise.}
    \end{cases} 
\end{split}
\end{align}
\begin{align}
\begin{split}
P_2(y)\defeq & \prod_{i=0}^{\left\lceil \log_2(m)\right\rceil} \left(1-y_i\right)\\
= & \begin{cases}
      1 & \text{if $\mathrm{y}=0$}\\
      0 & \text{otherwise, i.e., $\mathrm{y}>0$.}
    \end{cases}
\end{split}
\end{align}
\label{lem:insurjPens}
\end{lemma}
\begin{proof}
    Clearly, $min_{x,y} P_1(x,y) + P_2(y)=0$ $\iff$ $f(x)>0$ and thus $P_1(x,y) + P_2(y)>0$ $\iff$ $f(x)=0$.
\end{proof}

When denoting the surjectivity constraint from lemma \ref{lem:surSignTerm} as $f(x)$, we can see, that $f(x)<\left| V\right|$ for every $v_i \in V$. Therefore, we can use lemma \ref{lem:insurjPens} to formulate a penalty term for the surjectivity constraint at the expense of at most $\left| V\right| \lceil\log_{2}\left| V\right|\rceil$ ancillary qubits.

With these results, we are now ready to formulate the following PUBO penalty term for injectivity.

\begin{lemma}
$\overline{\sigma}_{ij}(x)$ is positive for every $v_i\in V$ and $v_j\in c(v_i)\setminus\left\lbrace v_i\right\rbrace$ not contained in the separation-node set iff the separation-node set associated with the $0$-flag vector $x\in\left\lbrace 0,1\right\rbrace ^{\left|V\right|}$ is injective, and $0$, otherwise.
\begin{equation}
\overline{\sigma}_{ij}(x)\defeq \sum_{\overline{p}\in p(v_i,v_j)} \prod_{k=1}^{dim(\overline{p})-1}\delta_{c(\overline{p}_k)c(\overline{p}_{k+1})}x_{\pi(\overline{p}_k)}x_{\pi(\overline{p}_{k+1})}
\end{equation}
\label{lem:inSignTerm}
\end{lemma}
\begin{proof}
    Here, $\overline{\sigma}_{ij}(x)$ is positive, iff a simple path $\overline{p}\in p(v_i,v_j)$ between $v_i$ and $v_j$ exists, that consists exclusively of nodes assigned to the community of $v_i$ which are not part of the separation-node set, and $0$, otherwise.
\end{proof}

Analogously to lemma \ref{lem:surSignTerm}, we can observe, that $\overline{\sigma}_{ij}(x)\leq 2^{\left| V\right| - 2}$. Thus, we can use lemma \ref{lem:insurjPens} at the expense of less than of $\lceil\log_{2}\left| V\right|\rceil$ ancillary qubits for every single node pair $v_i\in V$ and $v_j\in c(v_i)\setminus\left\lbrace v_i\right\rbrace$. As injectivity demands the positiveness of $\overline{\sigma}_{ij}(x)$ for all node pairs $v_i\in V$ and $v_j\in c(v_i)\setminus\left\lbrace v_i\right\rbrace$, $\left| V\right|^{2} \lceil\log_{2}\left| V\right|\rceil$ ancillary qubits suffice to construct a penalty term for injectivity. The selection of appropriate node pairs $v_i\in V$ and $v_j\in c(v_i)\setminus\left\lbrace v_i\right\rbrace$ can be done using the term $x_i x_j \delta_{c(v_i)c(v_j)}$.
%\balance

\end{document}